\documentclass[10pt,onecolumn]{article}
\usepackage[margin=1.5in,nohead]{geometry}
\usepackage{cite}
\usepackage{microtype}
\usepackage{amsmath}
\usepackage{amssymb}
\usepackage{amsthm}
\usepackage{algorithm}
\usepackage[noend]{algpseudocode}
\usepackage{caption}
\usepackage{subcaption}
\usepackage{booktabs}
\usepackage{url}
\newtheorem{theorem}{Theorem}
\newtheorem{lemma}{Lemma}

\newtheorem{proposition}{Proposition}
\newtheorem{observation}{Observation}
\newtheorem{definition}{Definition}
\newtheorem{claim}{Claim}
\usepackage{tikz}
\usetikzlibrary{arrows,positioning,trees}

\newcommand{\tikznode}[2]{%
\ifmmode%
\tikz[remember picture,baseline=(#1.base),inner sep=0pt] \node (#1) {$#2$};%
\else
\tikz[remember picture,baseline=(#1.base),inner sep=0pt] \node (#1) {#2};%
\fi}

\newcommand{\StrikeColumn}[3][]{%
  \begin{tikzpicture}[overlay,remember picture]
    \draw[#1] (#2.north) -- (#3.south);
  \end{tikzpicture}
}

\newcommand{\StrikeRow}[3][]{%
  \begin{tikzpicture}[overlay,remember picture]
    \draw[#1] (#2.west) -- (#3.east);
  \end{tikzpicture}
}

\DeclareMathOperator{\supp}{supp}
\renewcommand{\vec}[1]{\mathbf{#1}}

\title{Optimal algebraic Breadth-First Search for sparse graphs}
\author{Paul Burkhardt \footnote{Research Directorate, National Security Agency,
Fort~Meade, MD 20755. Email: pburkha@nsa.gov}}

\begin{document}

\onecolumn
\maketitle

\begin{abstract}
  There has been a rise in the popularity of algebraic methods for graph
  algorithms given the development of the GraphBLAS library and other sparse
  matrix methods. An exemplar for these approaches is Breadth-First Search
  (BFS). The algebraic BFS algorithm is simply a recurrence of matrix-vector
  multiplications with the $n \times n$ adjacency matrix, but the many redundant
  operations over nonzeros ultimately lead to suboptimal performance. Therefore
  an optimal algebraic BFS should be of keen interest especially if it is easily
  integrated with existing matrix methods.

  Current methods, notably in the GraphBLAS, use a Sparse Matrix masked-Sparse
  Vector (SpMmSpV) multiplication in which the input vector is kept in a sparse
  representation in each step of the BFS, and nonzeros in the vector are masked
  in subsequent steps. This has been an area of recent research in GraphBLAS and
  other libraries. While in theory these masking methods are asymptotically
  optimal on sparse graphs, many add work that leads to suboptimal runtime. We
  give a new optimal, algebraic BFS for sparse graphs, thus closing a gap in the
  literature.

  Our method multiplies progressively smaller submatrices of the adjacency
  matrix at each step. Let $n$ and $m$ refer to the number of vertices and
  edges, respectively. On a sparse graph, our method takes $O(n)$ algebraic
  operations as opposed to $O(m)$ operations needed by theoretically optimal
  sparse matrix approaches. Thus for sparse graphs it matches the bounds of the
  best-known sequential algorithm and on a Parallel Random Access Machine (PRAM)
  it is work-optimal. Our result holds for both directed and undirected
  graphs. Compared to a leading GraphBLAS library our method achieves up to 24x
  faster sequential time and for parallel computation it can be 17x faster on
  large graphs and 12x faster on large-diameter graphs.

  \textit{Keywords}: breadth-first search, graph algorithm, sparse matrix,
  linear algebra
\end{abstract}

\section{\label{sec:introduction}Introduction}
\emph{Breadth-First Search} (BFS) is a principal search algorithm and
fundamental primitive for many graph algorithms such as computing reachability
and shortest paths. Let $n$ and $m$ refer to the number of vertices and edges,
respectively. By labeling vertices $1..n$, a symmetric $n\times n$ adjacency
matrix, $A$, can be constructed so that every nonzero element of the matrix
denotes an edge leading to $O(m)$ nonzeros in total. Hence each column or row
vector of this matrix describes the adjacency or neighborhood of a vertex. The
linear algebraic BFS algorithm is simply a recurrence of matrix-vector
multiplications with this adjacency matrix and the previous multiplication
product. It solves the $\vec{x}_{k+1} = A\vec{x}_k$ relation, where each
matrix-vector product captures the next level in the search. The computation is
optimal if it makes $O(m)$ algebraic operations overall. Observe that this
recurrence can be iterated to give $\vec{x}_{n+1}=A^n\vec{x}_1$ and therefore
matrix exponentiation of $A$ by repeated squaring\footnote{Ex. $x^{17}=x\times
  x^{16} = x\times x^{4^{^4}} = x\times x^{2^{2^{2^{2}}}}$} leads to a
sublinear-time, parallel BFS, but requires $\Omega(n^3 \log n)$ work.

For sparse graphs, computing the algebraic BFS by $\vec{x}_{k+1} = A\vec{x}_k$
is appealing due in part to the availability of highly optimized matrix
libraries that are finely tuned to the computer architectures. These libraries
take advantage of the memory subsystem and it is this low-level interaction with
hardware that enables the algebraic BFS to be faster in practice than the
theoretically optimal combinatorial algorithm. Newer approaches employ a Sparse
Matrix masked-Sparse Vector (SpMmSpV) multiplication~\cite{bib:yang2018,
  bib:azad_buluc2017, bib:yang2015, bib:buluc_madduri2011}. In these methods,
previously visited frontier vertices are masked out of the sparse input vector
at each step. This can lead to an optimal BFS on sparse graphs, but makes more
algebraic operations than needed. Moreover, many of the new SpMmSpV-BFS methods
add work that degrade the performance to $O(mn)$ time.

We give an algebraic BFS that is optimal for sparse graphs. Our method
multiplies progressively smaller submatrices of the adjacency matrix at each
step. It masks both row and column vectors in $A$ so it avoids returning
nonzeros from the transpose of a column vector in subsequent calculations. Only
the column vectors are needed to project the next frontier so the row vectors
for the current frontier vertices are also masked during the multiplication on
that frontier. This exploits the symmetry of $A$ and thus requires at most $m$
nonzeros. Our new BFS also short-circuits each matrix-vector multiplication by
updating the mask within a step and thus references $n-1$ nonzeros in
total. This holds for both undirected and directed graphs. In contrast, an
optimal SpMmSpV-BFS, even with short-circuiting, references all $2m$ nonzeros in
the worst-case on an undirected graph. Therefore our approach requires
significantly fewer algebraic operations than a theoretically optimal
SpMmSpV-BFS method. Our submatrix multiplication method is quite simple so it is
surprising that it has been overlooked~\cite{bib:conf2015}. Our new method can
be easily integrated with existing matrix methods, and may benefit the masking
techniques in the GraphBLAS library~\cite{bib:davis2019, bib:yang2019,
  bib:yang2018, bib:buluc_graphblas2017}, so we expect it would provide
substantial value in practical settings.

Although our technique is applicable to dense graphs, it only leads to optimal
performance on sparse graphs. For the remainder of this paper we will only
consider sparse graphs. We summarize our contribution in
Section~\ref{sec:contribution}. A brief background on sparse matrix methods for
BFS is given in Section~\ref{sec:background}. We review the current SpMmSpV
approaches in Section~\ref{sec:related}. In Sections~\ref{sec:submatrix}
and~\ref{sec:bounds} we define a new algebraic BFS by submatrix multiplication
and analyze the asymptotic bounds on operations. Then in
Section~\ref{sec:algorithm} we describe our main algorithm and show that its
performance matches the combinatorial algorithm, and on a Parallel Random Access
Machine (PRAM) it is work-optimal. Then in Sections~\ref{sec:sequential}
and~\ref{sec:parallel} we use the popular Compressed Sparse Row (CSR) format to
demonstrate the theoretical contribution and how easily it can be integrated
with existing sparse matrix methods. Experimental results of this are given in
Section~\ref{sec:experiment} where we demonstrate faster performance than the
leading GraphBLAS library, achieving up to 24x faster sequential time and for
parallel computation it can be 17x faster on large graphs and 12x faster on
large-diameter graphs.

\section{\label{sec:notation}Notation}
All following descriptions are for simple, undirected graphs denoted by
$G=(V,E)$ with $n=\lvert V \rvert$ vertices and $m=\lvert E \rvert$ edges. The
number of neighbors of a vertex is given by its degree $d(v)=\lvert \{u | (u,v)
\in E \} \rvert$. Let $D_G$ denote the diameter of $G$, meaning the maximum
distance between any pair of vertices in $G$. The vertices in $G$ are labeled
$[n]=1,2,3...n$. Let $A \in \{0,1\}^{n \times n}$ be the adjacency matrix for
$G$. We use $A[\cdot, \cdot]$ to denote the submatrix of $A$ by its rows and
columns. For example, $A[\{1,2,4\},\{3,4\}]$ is the submatrix given by rows
$1,2,4$ and columns $3,4$ in $A$. We also use $A_{i,*},A_{*,i}$ for the
$i^{\text{th}}$ row and column vectors of $A$, respectively. The support of a
vector $\vec{x}$, denoted by $\supp(\vec{x})$, refers to the set of indices
corresponding to nonzero entries in $\vec{x}$, thus $\supp(\vec{x}) = \{i | x(i)
\ne 0\}$.

Correctness in the algebraic BFS requires only the distinction between zero and
nonzero values and this also holds in our method. In addition to the Arithmetic
semiring it is safe to use the Boolean (OR for addition, AND for multiplication)
or Tropical min-plus semiring, both of which also avoid bit-complexity
concerns. When appropriate, we'll denote the addition and multiplication
operators by the conventional symbols, $\oplus$ and $\otimes$, respectively. To
keep our discussion simple, we'll assume that any algebraic operation takes
$O(1)$ time.

\section{\label{sec:contribution}Our contribution}
We give a new algebraic BFS by submatrix multiplication that takes fewer
algebraic operations than a theoretically optimal SpMmSpV-BFS method. We denote
our algebraic BFS by the recurrence $\vec{x}_{k+1}=A[V_{k+1},V_k]\vec{x}_k$
where $A[V_{k+1},V_k]$ is the submatrix of $A$ with row and column indices given
by the set $V_k$ containing vertices not yet found in the search as of step
$k$. The vector $\vec{x}_k$ is also masked by the indices in $V_k$ so in all
steps the matrix and vector are compatible. We will show that our BFS takes
$O(n)$ algebraic operations as opposed to $O(m)$ operations of an optimal
SpMmSpV-BFS. Our results hold for both undirected and directed graphs, only
$A^T[V_{k+1},V_k]$ is used for directed graphs. Our main theoretical results are
given by the following theorems.

\newtheorem*{thm:abfs}{Theorem \ref{thm:abfs}}
\begin{thm:abfs}
  Breadth-First Search can be computed by
  $\vec{x}_{k+1}=A[V_{k+1},V_k]\vec{x}_k$ for $k=1,\ldots,O(n)$ steps.
\end{thm:abfs}

\newtheorem*{thm:abfs_ops}{Theorem \ref{thm:abfs_ops}}
\begin{thm:abfs_ops}
  Computing Breadth-First Search by $\vec{x}_{k+1}=A[V_{k+1},V_k]\vec{x}_k$
  using one nonzero in each row will multiply $n-1$ nonzeros in $A$.
\end{thm:abfs_ops}

We introduce a new algebraic BFS in Algorithm~\ref{alg:abfs_spmv} based on
Theorems~\ref{thm:abfs} and~\ref{thm:abfs_ops}. The adjacency matrix remains
unchanged, rather we are masking the rows and columns in the matrix that
corresponds to previously visited vertices. The input vector in each step is
also effectively masked so it is a sparse vector. Hence our method multiplies a
sparse submatrix by a sparse subvector in decreasing size each step. This leads
to an asymptotic speedup over the conventional algebraic method for both
sequential and parallel computation. Our algorithm is optimal on sparse graphs
and work-optimal on a PRAM. Our main algorithmic results are the following.

\newtheorem*{thm:abfs_spmv}{Theorem \ref{thm:abfs_spmv}}
\begin{thm:abfs_spmv}
  Algorithm~\ref{alg:abfs_spmv} computes an algebraic Breadth-First Search in
  $O(m+n)$ time for sparse $G$.
\end{thm:abfs_spmv}

\newtheorem*{thm:abfs_pram}{Theorem \ref{thm:abfs_pram}}
\begin{thm:abfs_pram}
  Algorithm~\ref{alg:abfs_spmv} computes an algebraic Breadth-First Search over
  $t$ steps in $O(t)$ time and $O(m)$ work using $O(m/t)$ PRAM processors for
  sparse $G$.
\end{thm:abfs_pram}

The sequential and parallel versions of this algorithm are deterministic and
asymptotically optimal for any ordering of matrix and vector indices. The
current state-of-the-art SpMmSpV-BFS approaches are only optimal if the vector
indices are unordered~\cite{bib:yang2018, bib:azad_buluc2017}. It also appears
that other recent SpMmSpV methods take $O(mn)$ time overall for BFS because
their masking method requires an elementwise multiplication with a dense vector
or explicitly testing every vertex in each step~\cite{bib:yang2018,
  bib:yang2015, bib:buluc_madduri2011}.

\section{\label{sec:background}Background}
The sequential combinatorial algorithm for Breadth-First Search is
well-known~\cite{bib:cormen2009} and attributed to the 1959 discovery by
Moore~\cite{bib:moore1959}. The algorithm proceeds iteratively where in each
step it finds the neighbors of vertices from the previous step such that these
neighbors are also unique to the current step. Each step constructs a set of
vertices that are not in other steps, known as the frontier set. These sets are
the levels of the BFS tree and so the traversal is called level-synchronous. The
algorithm takes $O(m+n)$ time due to referencing $O(n)$ vertices and testing
$O(m)$ edges. To avoid cycles the algorithm must proceed one level at a
time. Since the graph diameter is bounded by $n$, then there are $D_G$ levels,
hence the search is inherently sequential. As of yet, there is no
sublinear-time, parallel algorithm for BFS that achieves $O(m+n)$ work.

Not long after Moore's discovery, Floyd and Warshall used the adjacency matrix
$A$ to solve transitive closure and shortest-path problems~\cite{bib:floyd1962,
  bib:warshall1962} which are generalizations of Breadth-First Search. The
linear algebraic BFS algorithm is simply the $\vec{x}_{k+1} = A\vec{x}_k$
recurrence. Suppose that $A$ is a sparse matrix, then computing an algebraic BFS
by $\vec{x}_{k+1} = A\vec{x}_k$ takes $O(mn)$ time because all $O(m)$ nonzeros
in $A$ are multiplied in all $O(n)$ steps. This Sparse Matrix-Vector (SpMV)
approach is clearly wasteful because nonzeros in $A$ will be multiplied by zeros
in $\vec{x}$. The situation is not improved using a Sparse Matrix Sparse Vector
(SpMSpV) multiplication where both the input vector and the matrix are in sparse
format. Applying SpMSpV in BFS can still take $\Omega(mD_G)$ time because
nonzeros reappear in $\vec{x}$, meaning every vertex is visited again, and so
$\vec{x}$ becomes dense. This is readily observed given a long path connected to
a clique where the search starts with a vertex in this clique. Hence it is not
enough to treat both $A$ and $\vec{x}$ as sparse, and consequently a
straightforward SpMSpV method for BFS is not optimal. But practical
implementations have provided speedup over the combinatorial
algorithm~\cite{bib:besta2017, bib:bucker2014, bib:beamer2012,
  bib:buluc_madduri2011, bib:kepner_gilbert2011}. These practical
implementations rely on the level-synchronous sequential algorithm where the
focus is on parallelizing the work within a level of the BFS tree. Newer methods
that mask frontier vertices in subsequent input vectors to avoid revisiting
vertices will be the subject of our review.

\section{\label{sec:related}Related work}
On average there are $\bar d = 2m/n$ nonzeros in each column of $A$. In a single
SpMSpV multiplication there are $f$ nonzeros in the sparse vector and thus
$\Omega(\bar{d} f)$ operations on average. Accessing a total of $f$ unique
columns in $A$ over all BFS steps then leads to $O(m)$ runtime, and is therefore
optimal for sparse graphs. This can be achieved by SpMmSpV methods that hide or
mask previously seen frontier nonzeros from the sparse input vector at each BFS
step. But many of the current approaches add more work that degrades the
runtime. In the following review of current SpMSpV and SpMmSpV algorithms for
BFS, we ignore any work related to initialization, parallelization, or other
overhead that do not affect the asymptotic complexity.

In an algebraic BFS on sparse graphs, the nonzeros from the multiplication must
be written to a new sparse output vector. Then using SpMSpV for the algebraic
BFS requires a multi-way merge due to the linear combination of either rows or
columns of $A$ that are projected by the nonzeros in the sparse vector in the
multiplication. Strategies for efficient merging include using a priority queue
(heap) or sorting, but this results in $\Omega(n\log n)$ runtime for
BFS~\cite{bib:yang2015}. Another popular method is to employ a sparse
accumulator (SPA)~\cite{bib:azad_buluc2017, bib:gilbert1992} which is comprised
of a dense vector of values, a dense vector of true/false flags, and an
unordered list of indices to nonzeros in the dense vector. This SPA can be used
to mask out previous frontier nonzeros. But it is stated
in~\cite{bib:azad_buluc2017} that there is no known algorithm for SpMSpV that
attains the lower-bound of $\Omega(\bar{d} f)$ if the indices in the sparse
vector must be sorted. This is because the list of row indices in the SPA must
be sorted if the sparse matrix was stored with ordered indices and the
multiplication algorithm requires that ordering~\cite{bib:gilbert1992}. Thus
SpMmSpV methods using a SPA take $\Omega(n\log n)$ time if the output vector
needs sorted indices, making their use in a BFS non-optimal.

The focus of these new SpMSpV methods is in efficiently reading and writing the
sparse vector. But there is an analysis gap on the asymptotic cost of preventing
previous frontier vertices in the BFS from reappearing in the sparse
vector. Masking out these frontier nonzeros was analyzed in~\cite{bib:yang2018}
and it appears to require an elementwise multiplication with a dense masking
vector, which must be $O(n)$ size to accommodate all vertices. This suggests
these SpMmSpV methods with masking take $O(mn)$ time for
BFS. In~\cite{bib:buluc_madduri2011} an elementwise multiplication with a dense
predecessor array is performed in each step of the BFS to mask the old frontier,
leading to $O(mn)$ runtime. The SpMmSpV method in~\cite{bib:buluc_madduri2011}
also required sorted output so either method of a priority queue or SPA leads to
suboptimal time. The SpMmSpV algorithm for BFS in~\cite{bib:yang2015} tests all
vertices in each step and zeros out those in the output vector that have already
been reached, leading to $\Omega(mn)$ time. A masked, column-based matrix-vector
method for BFS that relies on radix sorting is given in~\cite{bib:yang2018} but
takes $\Omega(m\log n)$ time. The authors allow unsorted indices to avoid the
$\Omega(\log n)$ factor but elementwise multiplication with the dense masking
vector results in $O(mn)$ time. Sorted vectors are also used
in~\cite{bib:azad_buluc2017}, thereby taking $\Omega(m\log n)$ time for BFS. A
version with unsorted indices is given in~\cite{bib:azad_buluc2017} but the
authors do not describe how visited vertices are avoided or masked.

Although our focus is on an optimal matrix-based BFS, we highlight a popular
method known as Direction-Optimizing BFS~\cite{bib:beamer2012}. The
Direction-Optimizing BFS is not asymptotically optimal, but works well in
practice for low-diameter, real-world graphs. For steps where the frontier is
large, the algorithm in~\cite{bib:beamer2012} switches to a ``bottom-up''
evaluation where the neighbors of each unvisited vertex, rather than a frontier
vertex, is scanned until a visited neighbor is found and thus the unvisited
vertex is marked visited. In practice this could perform fewer edge checks
because it stops on the first visited neighbor. But suppose a small fraction of
edges without a visited endpoint are checked during the
``bottom-up''phase. These edges could be in the same component of the BFS
traversal but are still some steps away from being reached, or they are in
different components and will not be reached. If this fraction is repeatedly
checked as a function of the input, it leads to sub-optimal runtime. For
example, if $\frac{m}{16}$ edges are checked as little as $\log n$ times, it
leads to $\Omega(m\log n)$ runtime.

\section{\label{sec:submatrix}Submatrix multiplication}
Recall that matrix-vector multiplication by the outer-product is the linear
combination of the column vectors in the matrix scaled by the entries in the
input vector as follows.

\begin{displaymath}
\begin{bmatrix}
0 & 1 & 1 & 0 & 0 \\
1 & 0 & 1 & 1 & 0 \\
1 & 1 & 0 & 0 & 0 \\
0 & 1 & 0 & 0 & 1 \\
0 & 0 & 0 & 1 & 0 \\
\end{bmatrix}
\begin{bmatrix}
0 \\
1 \\
1 \\
0 \\
0
\end{bmatrix}
=
\begin{bmatrix}
1 \\
0 \\
1 \\
1 \\
0
\end{bmatrix}
\oplus
\begin{bmatrix}
1 \\
1 \\
0 \\
0 \\
0
\end{bmatrix}
=
\begin{bmatrix}
1 \\
1 \\
1 \\
1 \\
0
\end{bmatrix}
\end{displaymath}

The neighbors of vertex $i$ are the nonzero elements in the $A_{*,i}$ column
vector of $A$. Here we show the dense matrix and vector for illustration only,
so the reader should keep in mind that zeros are ignored in the calculations
including those in the vectors. The linear combination of the $A_{*,2}$ and
$A_{*,3}$ columns result in the nonzeros at $1,2,3,4$ indices of the product
vector. In BFS this corresponds to finding the neighbors $1,3,4$ of vertex $2$
and neighbors $1,2$ of vertex $3$. The search continues by multiplying the
matrix with this new product vector. But in the next step the $A_{*,2}$ and
$A_{*,3}$ columns are projected again resulting in redundant operations that do
not add new vertices to the search. This will result in each vertex being
revisited leading to $O(mn)$ time. Masking nonzeros in $\vec{x}$ prevents their
recurrence and leads to $O(m)$ optimal time. But masking the vector alone still
incurs twice as many algebraic operations than theoretically needed. For
example, merely ignoring the $2,3$ elements in this next input vector does not
eliminate their recurrence because the columns $A_{*,1},A_{*,4}$ will give $2,3$
again in the following output vector. A theoretically optimal SpMmSpV method for
BFS will make $\sum_v d(v)= O(m)$ additional operations.

We notice that the computation can be performed over progressively smaller
submatrices of $A$ so at most half the nonzero elements in $A$ are needed. This
is because each unique frontier nonzero over all steps can be produced by a
single $A(i,j)$ element. As described earlier, nonzeros can reappear in the
input vector in later steps because they are found in the transpose of their
column vectors. Then after each step $k$, masking the transpose elements
$A_{j,*},A_{*,j}$ and $x_k(j)$ for each nonzero $j^{\text{th}}$ index of
$\vec{x}_k$ will ignore these nonzeros for all remaining steps. Moreover, only
column vectors are needed to progress the search so the $A_{j,*}$ row vectors
can be masked at step $k$, one step earlier than their transpose
counterparts. This eliminates returning frontier vertices that are
adjacent. Therefore the algebraic Breadth-First Search can be computed over
submatrices of $A$ in which the row and column indices that remain are those not
used in preceding steps. This is illustrated in
Figure~\ref{fig:abfs}. Consequently the number of algebraic operations are
reduced. For sparse graphs the algebraic method will now be as efficient as the
combinatorial BFS.

We emphasize that $A$ can be left unchanged by masking the appropriate
submatrices of $A$. It is also implicit in the following descriptions that the
submatrices of $A$ and the input vectors $\vec{x}$ are compatible at each step
because the column span of $A$ always matches the row dimension of the input
vector.

\begin{figure}[t]
\centering
\small
\resizebox{.35\textwidth}{!}{%
\begin{tikzpicture}
[scale=.1,
vertex/.style={circle,draw=black,thin},
edge/.style={thin}]
\node [vertex] (2) {2};
\node [vertex] (5) [right=of 2] {5};
\node [vertex] (3) [right=of 5] {3};
\node [vertex] (4) [right=of 3] {4}; 
\node [vertex] (1) [right=of 4] {1}; 
\draw [edge] (2) to (5);
\draw [edge] (5) to (3);
\draw [edge] (3) to (4);
\draw [edge] (4) to (1);
\end{tikzpicture}
}
\\
\resizebox{.35\textwidth}{!}{%
\begin{minipage}{.25\textwidth}
\begin{align*}
\begin{array}{c}
\begin{matrix} 1 & 2 & 3 & 4 & 5
\end{matrix} \\
\begin{matrix}
1 \\
2 \\
3 \\
4 \\
5
\end{matrix}
\begin{bmatrix}
0 & 0 & 0 & 1 & 0 \\
\tikznode{A2S1Left}{0} & 0 & 0 & 0 & \tikznode{A2S1Right}{1} \\
0 & 0 & 0 & 1 & 1 \\
1 & 0 & 1 & 0 & 0 \\
0 & 1 & 1 & 0 & 0
\end{bmatrix}
\end{array}
\begin{array}{c}
\begin{matrix} \end{matrix} \\
\begin{bmatrix}
0 \\
1 \\
0 \\
0 \\
0
\end{bmatrix}
\end{array}
&=
\begin{array}{c}
\begin{matrix} \end{matrix} \\
\begin{bmatrix}
0 \\
\tikznode{C2S1}{0} \\
0 \\
0 \\
1
\end{bmatrix}
\end{array}
& \text{Step 1} \\
\begin{array}{c}
\begin{matrix} 1 & 2 & 3 & 4 & 5
\end{matrix} \\
\begin{matrix}
1 \\
2 \\
3 \\
4 \\
5
\end{matrix}
\begin{bmatrix}
0 & \tikznode{A2S2Top}{0} & 0 & 1 & 0 \\
\tikznode{A2S2Left}{0} & 0 & 0 & 0 & \tikznode{A2S2Right}{1} \\
0 & 0 & 0 & 1 & 1 \\
1 & 0 & 1 & 0 & 0 \\
\tikznode{A5S2Left}{0} & \tikznode{A2S2Bottom}{1} & 1 & 0
& \tikznode{A5S2Right}{0}
\end{bmatrix}
\end{array}
\begin{array}{c}
\begin{matrix} \end{matrix} \\
\begin{bmatrix}
0 \\
\tikznode{B2S2}{0} \\
0 \\
0 \\
1
\end{bmatrix}
\end{array}
&=
\begin{array}{c}
\begin{matrix} \end{matrix} \\
\begin{bmatrix}
0 \\
\tikznode{C2S2}{0} \\
1 \\
0 \\
\tikznode{C5S2}{0} \\
\end{bmatrix}
\end{array}
& \text{Step 2} \\
\begin{array}{c}
\begin{matrix} 1 & 2 & 3 & 4 & 5
\end{matrix} \\
\begin{matrix}
1 \\
2 \\
3 \\
4 \\
5
\end{matrix}
\begin{bmatrix}
0 & \tikznode{A2S3Top}{0} & 0 & 1 & \tikznode{A5S3Top}{0} \\
\tikznode{A2S3Left}{0} & 0 & 0 & 0 & \tikznode{A2S3Right}{1} \\
\tikznode{A3S3Left}{0} & 0 & 0 & 1 & \tikznode{A3S3Right}{1} \\
1 & 0 & 1 & 0 & 0 \\
\tikznode{A5S3Left}{0} & \tikznode{A2S3Bottom}{1} & 1 & 0 &
\tikznode{A5S3Bottom}{0}
\end{bmatrix}
\end{array}
\begin{array}{c}
\begin{matrix} \end{matrix} \\
\begin{bmatrix}
0 \\
\tikznode{B2S3}{0} \\
1 \\
0 \\
\tikznode{B5S3}{0} \\
\end{bmatrix}
\end{array}
&=
\begin{array}{c}
\begin{matrix} \end{matrix} \\
\begin{bmatrix}
0 \\
\tikznode{C2S3}{0} \\
\tikznode{C3S3}{0} \\
1 \\
\tikznode{C5S3}{0} \\
\end{bmatrix}
\end{array}
& \text{Step 3} \\
\begin{array}{c}
\begin{matrix} 1 & 2 & 3 & 4 & 5
\end{matrix} \\
\begin{matrix}
1 \\
2 \\
3 \\
4 \\
5
\end{matrix}
\begin{bmatrix}
0 & \tikznode{A2S4Top}{0} & \tikznode{A3S4Top}{0} & 1 & \tikznode{A5S4Top}{0} \\
\tikznode{A2S4Left}{0} & 0 & 0 & 0 & \tikznode{A2S4Right}{1} \\
\tikznode{A3S4Left}{0} & 0 & 0 & 1 & \tikznode{A3S4Right}{1} \\
\tikznode{A4S4Left}{1} & 0 & 1 & 0 & \tikznode{A4S4Right}{0} \\
\tikznode{A5S4Left}{0} & \tikznode{A2S4Bottom}{1} & \tikznode{A3S4Bottom}{1} & 0
& \tikznode{A5S4Bottom}{0}
\end{bmatrix}
\end{array}
\begin{array}{c}
\begin{matrix} \end{matrix} \\
\begin{bmatrix}
0 \\
\tikznode{B2S4}{0} \\
\tikznode{B3S4}{0} \\
1 \\
\tikznode{B5S4}{0} \\
\end{bmatrix}
\end{array}
&=
\begin{array}{c}
\begin{matrix} \end{matrix} \\
\begin{bmatrix}
1 \\
\tikznode{C2S4}{0} \\
\tikznode{C3S4}{0} \\
\tikznode{C4S4}{0} \\
\tikznode{C5S4}{0} \\
\end{bmatrix}
\end{array}
& \text{Step 4}
\end{align*}
\StrikeRow[black, thick]{A2S1Left}{A2S1Right}
\StrikeRow[black, thick]{C2S1}{C2S1}
\StrikeColumn[black, thick]{A2S2Top}{A2S2Bottom}
\StrikeRow[black, thick]{A2S2Left}{A2S2Right}
\StrikeRow[black, thick]{A5S2Left}{A5S2Right}
\StrikeRow[black, thick]{B2S2}{B2S2}
\StrikeRow[black, thick]{C2S2}{C2S2}
\StrikeRow[black, thick]{C5S2}{C5S2}
\StrikeColumn[black, thick]{A2S3Top}{A2S3Bottom}
\StrikeColumn[black, thick]{A5S3Top}{A5S3Bottom}
\StrikeRow[black, thick]{A2S3Left}{A2S3Right}
\StrikeRow[black, thick]{A3S3Left}{A3S3Right}
\StrikeRow[black, thick]{A5S3Left}{A5S3Bottom}
\StrikeRow[black, thick]{B2S3}{B2S3}
\StrikeRow[black, thick]{B5S3}{B5S3}
\StrikeRow[black, thick]{C2S3}{C2S3}
\StrikeRow[black, thick]{C3S3}{C3S3}
\StrikeRow[black, thick]{C5S3}{C5S3}
\StrikeColumn[black, thick]{A2S4Top}{A2S4Bottom}
\StrikeColumn[black, thick]{A3S4Top}{A3S4Bottom}
\StrikeColumn[black, thick]{A5S4Top}{A5S4Bottom}
\StrikeRow[black, thick]{A2S4Left}{A2S4Right}
\StrikeRow[black, thick]{A3S4Left}{A3S4Right}
\StrikeRow[black, thick]{A4S4Left}{A4S4Right}
\StrikeRow[black, thick]{A5S4Left}{A5S4Bottom}
\StrikeRow[black, thick]{B2S4}{B2S4}
\StrikeRow[black, thick]{B3S4}{B3S4}
\StrikeRow[black, thick]{B5S4}{B5S4}
\StrikeRow[black, thick]{C2S4}{C2S4}
\StrikeRow[black, thick]{C3S4}{C3S4}
\StrikeRow[black, thick]{C4S4}{C4S4}
\StrikeRow[black, thick]{C5S4}{C5S4}
\end{minipage}
}
\caption{\label{fig:abfs}
Optimal algebraic BFS starting from vertex 2.}
\end{figure}
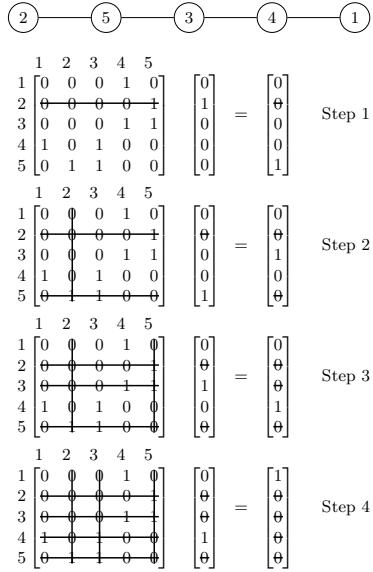

\begin{definition}
  Let $V_k$ be the set of remaining indices to vertices not yet visited by step
  $k$ of the BFS. That is $V_k = V_{k-1} \setminus \supp(\vec{x}_{k-1})$, where
  $V_0$ contains $\{1,2,\ldots,n\}$ and $\vec{x}_0=\vec{0}$.
\end{definition}

\begin{theorem}
  \label{thm:abfs}
  Breadth-First Search can be computed by
  $\vec{x}_{k+1}=A[V_{k+1},V_k]\vec{x}_k$ for $k=1,\ldots,O(n)$ steps.
\end{theorem}

\begin{proof}
  Recall the matrix-vector outer-product is a linear combination of column
  vectors in $A$,

  \begin{displaymath}
    \vec{x}_{k+1} = \bigoplus_j x_k(j)\otimes A_{*,j}.
  \end{displaymath}

  Only nonzeros in $\vec{x}_k$ can produce nonzeros in the resultant vector
  $\vec{x}_{k+1}$ because any nonzero $x_{k+1}(i)$ vector element is due to a
  nonzero $x_k(j)\otimes A(i,j)$ product. Also observe that a $A_{*,j}$ column
  vector can only produce a nonzero $x_{k+1}(i)$ if $A(i,j)$ is nonzero. Thus,
  subsequent operations on the $A_{*,j}$ column vector do not produce new
  $x_{k+1}(i)$ nonzeros. For Breadth-First Search this does not update a new
  level. Then for each $j$ in $\supp(\vec{x}_k)$ at step $k$, the
  $A_{j,*},A_{*,j}$ can be ignored in all remaining steps. Since only column
  vectors are used in the linear combination then each $A_{j,*}$ row vector can
  also be ignored at step $k$. Hence leading to
  $\vec{x}_{k+1}=A[V_{k+1},V_k]\vec{x}_k$ as claimed.
\end{proof}

\begin{claim}
  \label{clm:abfs_column}
  Each column vector in $A$, and subsequently each element of $A$, is multiplied
  at most once in computing Breadth-First Search by Theorem~\ref{thm:abfs}.
\end{claim}

\begin{proof}
  Only nonzeros are required in Breath-First Search so the matrix-vector product
  in Theorem~\ref{thm:abfs} is computed from only the $j$ indices in $V_k$ by

  \begin{displaymath}
    x_{k+1}(i) = \bigoplus_{j\in V_k \cap \supp(\vec{x}_k)}
    x_k(j) \otimes A(i,j).
  \end{displaymath}

  Now for all remaining steps since $V_k$ does not contain indices from
  $V_{k-1}$ then the multiplication over the submatrix $A[V_{k+1},V_k]$ does not
  include $A_{*,j},A_{j,*}$. Then there cannot be a vertex $i$ at some later
  step that produces $j$ by $x_{k+1}(j) = x_k(i) \otimes A(j,i)$ because all
  $A_{j,*}$ are prohibited. Thus each $A_{*,j}$ column vector can be multiplied
  only once and subsequently any element in $A$ is accessed no more than once.
\end{proof}

Notice it is valid to compute BFS by $\vec{x}_{k+1}=A[V_k,V_k]\vec{x}_k$ where
the submatrices are symmetric in each step, and still remain asymptotically
optimal. Using this symmetric form makes some of the theoretical results more
convenient. For example, we can show there is a linear transformation on the
conventional algebraic recurrence that produces our output. Namely, we can
derive $\vec{x}_{k+1}=A[V_k,V_k]\vec{y}_k$ where we use
$\vec{y}_k=A\vec{y}_{k-1}$ to denote the conventional BFS recurrence. This leads
to $\vec{x}_{k+1}=A[V_k,V_k]A^{k-1}\vec{x}_1$.

\begin{proposition}
  \label{prop:linear_transformation}
  There is a linear transformation on $\vec{y}_k = A\vec{y}_{k-1}$ that gives
  $\vec{x}_{k+1}=A[V_k,V_k]\vec{x}_k$, specifically $\vec{x}_{k+1} =
  A_k\vec{y}_k$ and subsequently $\vec{x}_{k+1}=A_kA^{k-1}\vec{x}_1$.
\end{proposition}

Thus we can show the equality between our method and the conventional
recurrence. An interested reader can refer to the
Appendix~\ref{sec:lineartransform} for the proof.

\section{\label{sec:bounds}Bounds on algebraic operations}
The submatrices $A[V_{k+1},V_k]$ strictly decrease in size as the search
progresses and so the number of algebraic operations also decreases. A simple
analysis will show that our submatrix approach accesses at most $m$ nonzeros on
any sparse graph. Given flexibility on updating the mask, it only requires $n-1$
nonzeros and therefore we can design an algorithm that takes $O(n)$ algebraic
operations as opposed $O(m)$ operations of an optimal SpMmSpV-BFS. Moreover,
this $O(n)$ bound holds for both undirected and directed graphs, demonstrating
that our technique has significant advantages over that of an optimal
SpMmSpV-BFS.

It is easy to see that an optimal SpMmSpV-BFS method multiplies $2m$ nonzeros in
$A$ because each column of $A$, and therefore every nonzero in $A$, is used in
the computation. Our method takes half as many operations.

\begin{claim}
  \label{clm:abfs_transpose}
  The transpose of any $A(i,j)$ multiplied in
  $\vec{x}_{k+1}=A[V_{k+1},V_k]\vec{x}$ will not be multiplied at any step.
\end{claim}

\begin{proof}
  At each step $k$, the column vectors $A_{*,j}$ for $j\in \supp(\vec{x}_k)$ are
  multiplied. Then these column vectors and their transpose are masked in all
  remaining steps. Thus any $A(i,j)$ element multiplied at step $k$, and its
  transpose element $A(j,i)$, will not be accessed in later steps. Moreover, the
  $A(j,i)$ elements are not multiplied in step $k$ because the $A_{j,*}$ row
  vectors are also masked at step $k$. Therefore the transpose of the $A(i,j)$
  matrix elements multiplied in computing $\vec{x}_{k+1}=A[V_{k+1},V_k]\vec{x}$
  are not used.
\end{proof}

\begin{lemma}
  \label{lem:abfs_nnz}
  Computing Breadth-First Search by $\vec{x}_{k+1}=A[V_{k+1},V_k]\vec{x}_k$
  multiplies at most $m$ nonzeros in $A$.
\end{lemma}

It follows from Claim~\ref{clm:abfs_transpose} that Lemma~\ref{lem:abfs_nnz}
holds. Now observe that only one nonzero in each row of $A$ is required to
produce a correct output in $\vec{x}$. Hence at most $n$ nonzeros in $A$ are
needed to compute BFS by this method. It also does not require a symmetric
matrix and therefore holds for both undirected and directed graphs. This leads
to only $O(n)$ algebraic operations.

\begin{theorem}
  \label{thm:abfs_ops} Computing Breadth-First Search by
  $\vec{x}_{k+1}=A[V_{k+1},V_k]\vec{x}_k$ using one nonzero in each row will
  multiply $n-1$ nonzeros in $A$.
\end{theorem}

\begin{proof}
  Claim~\ref{clm:abfs_transpose} establishes that the transpose of $A(i,j)$
  elements used in computing BFS by Theorem~\ref{thm:abfs} are not
  multiplied. Moreover, only one nonzero from a row in $A$ is needed by the
  computation. It follows from Claim~\ref{clm:abfs_column} that columns are not
  repeated, hence only unique $A(i,j)$ elements can produce the frontier
  vertices. Excluding the source vertex, this leads to a total of $n-1$ nonzeros
  in $A$ in the multiplication.
\end{proof}

We remark that short-circuiting, as proposed in Theorem~\ref{thm:abfs_ops}, can
also be applied to SpMmSpV-BFS methods, meaning the evaluation of the semiring
to produce a new frontier nonzero stops when the result is not zero. But this
still leads to more algebraic operations than our approach.

\begin{observation}
  \label{obs:shortcircuit}
  An optimal SpMmSpV-BFS using one nonzero in each row of $A$ saves between $0$
  and $2(m-n)$ algebraic operations, and therefore in total requires between
  $2n$ and $2m$ algebraic operations.
\end{observation}

\begin{proof}
  If a column vector is projected by the matrix-vector multiplication at step
  $k$ then its transpose will appear in the linear combination at step
  $k+1$. Short-circuiting the semiring evaluation on this transpose, now a row
  vector, at step $k+1$ will evaluate just one nonzero. But not all nonzeros
  from the column vector at step $k$ will be multiplied at step $k+1$ because
  some will have been previous frontier nonzeros. In the worst case, such as a
  path, there are no savings from short-circuiting. At best the short-circuiting
  can save $d(v)-2$ operations in each row because at least one nonzero must be
  a parent from a previous frontier and another evaluates to a new frontier
  nonzero, leading to $\sum_{v\in V}(d(v)-2) = 2(m-n)$ fewer evaluations. Hence
  short-circuiting saves between $0$ and $2(m-n)$ algebraic operations as
  claimed. Then the total number of algebraic operations after short-circuiting
  is between $2m-(2(m-n))= 2n$ and $2m$. Therefore an optimal SpMmSpV-BFS with
  short-circuiting requires between $2n$ and $2m$ algebraic operations.
\end{proof}

We see from Observation~\ref{obs:shortcircuit} that an optimal SpMmSpV-BFS with
short-circuiting requires at best $2n$ algebraic operations in the most
optimistic scenario, and in the worst case no semiring evaluations are
short-circuited and hence it multiplies all $2m$ nonzeros. In contrast, our
approach takes at most $n-1$ nonzeros as asserted by Theorem~\ref{thm:abfs_ops}.

Let us briefly explore our submatrix multiplication on directed graphs. Recall
the conventional linear algebra BFS is given by the recurrence $\vec{x}_{k+1} =
A\vec{x}_k$, with the implicit assumption that the underlying graph is
undirected. With a directed graph the adjacency matrix $A$ is not symmetric. By
convention, egress and ingress edges are indicated as nonzeros in the rows and
columns of $A$, respectively. Since BFS traverses the egress edges of frontier
vertices, we use $\vec{x}_{k+1} = A^T\vec{x}_k$. Now observe that the
outer-product for the matrix-vector multiplication proceeds as usual by
projecting the column vectors from $A^T$ that correspond to nonzeros in
$\vec{x}_k$. Since our method does not mask out the column vectors for nonzeros
in $\vec{x}_k$ until the next step, then it works on both undirected and
directed graphs.

\section{\label{sec:algorithm}Optimal algorithm}
We now give our main algorithm for sparse graphs that employs the new method in
Theorem~\ref{thm:abfs}. Our new Algorithm~\ref{alg:abfs_spmv} does not specify a
sparse matrix format to be as general as possible.

\begin{algorithm}[H]
\caption{\label{alg:abfs_spmv}}
\begin{algorithmic}[1]
\Require $A$ \Comment adjacency matrix in sparse representation
\Statex Initialize $V_1$ with $1,2,..n$
\Statex Initialize $\vec{x}_1$ with the source
\Statex Set $V_2 := V_1 \setminus \supp(\vec{x}_1)$
\For{$k=1,2,\ldots$ until end of component}
  \ForAll{$j \in \supp(\vec{x}_k)$}
    \State set $V^{\prime} := V_{k+1}$
    \ForAll{$i \in \supp(A[V^{\prime},j])$}
      \State $x_{k+1}(i) \gets x_{k+1}(i) \oplus x_k(j) \otimes A(i,j)$
      \State\label{line:mask} set $V^{\prime} := V^{\prime}\setminus \{i\}$
      \Comment Achieves $V_{k+2} = V_{k+1} \setminus \supp(\vec{x}_{k+1})$
    \EndFor
  \EndFor
\EndFor
\end{algorithmic}
\end{algorithm}

\begin{theorem}
  \label{thm:abfs_spmv}
  Algorithm~\ref{alg:abfs_spmv} computes an algebraic Breadth-First Search in
  $O(m+n)$ time for sparse $G$.
\end{theorem}

\begin{proof}
  The algorithm computes Breadth-First Search by Theorem~\ref{thm:abfs}. At each
  step $k$ only the nonzero indices in $\vec{x}_k$ and $A$ are involved in the
  $A[V_{k+1},V_k]\vec{x}$ submatrix multiplication. Each $j$ corresponding to
  nonzeros in $\vec{x}_k$ projects a $A[V_{k+1},j]$ masked column vector in the
  multiplication. The subscripts $i$ corresponding to nonzeros in $A[V_{k+1},j]$
  produce the next $\vec{x}_{k+1}$ frontier nonzeros. After operating on an
  $A(i,j)$ element during step $k$, the $i^{th}$ row in $A$ is masked by
  removing $i$ from $V_{k+1}$ (line \ref{line:mask}). This ensures only one
  nonzero in a row is multiplied and at the end of the step the $V_{k+2}$ has
  been realized. Therefore any $A(i,j)$ multiplied at a step $k$ will not be
  used in subsequent calculations, and by Claim~\ref{clm:abfs_transpose} its
  transpose will not be multiplied. Following Theorem~\ref{thm:abfs_ops} there
  are $O(n)$ algebraic operations.

  The total time is as follows. There are $O(n)$ entries in total added to all
  $\vec{x}_k$, taking $O(n)$ time. This is possible by storing each new $i$ from
  the inner loop in a separate array that can be iterated over in each
  step. There are $O(m)$ nonzeros in $A$ and thus $O(m)$ tests to determine the
  $i$ in all $V_{k+1}$ that are ultimately masked. This takes $O(m)$ time,
  possibly by storing each $i$ in another array of size $n$ so in $O(1)$ time it
  can be determined if a $i$ has been previously used, and then marked as
  visited. Each algebraic operation takes $O(1)$ time so all algebraic
  operations takes $O(n)$ time. Therefore in total it takes $O(m+n)$ time.
\end{proof}

Algorithm~\ref{alg:abfs_spmv} is also work-optimal on a PRAM under the
Work-Depth (WD) model~\cite{bib:vishkin1982, bib:jaja1992, bib:blelloch1996}. In
this model a \textbf{for all} construct denotes a parallel region in which
instructions are performed concurrently. All other statements outside this
construct are sequential. The scheduling of processors is handled implicitly and
an arbitrary number of simultaneous operations can be performed each step. The
work is the sum of actual operations performed and the number of processors is
not a parameter in the WD model. The depth is the longest chain of dependencies,
often the number of computation steps, and is denoted by $D$. The
work-complexity, $W$, denotes the total number of operations. A parallel
algorithm in WD can be simulated on a $p$-processor PRAM in $O(W/p + D)$ time.
Using $p=W/D$ processors achieves optimal work on the PRAM.

\begin{theorem}
  \label{thm:abfs_pram}
  Algorithm~\ref{alg:abfs_spmv} computes an algebraic Breadth-First Search over
  $t$ steps in $O(t)$ time and $O(m)$ work using $O(m/t)$ PRAM processors for
  sparse $G$.
\end{theorem}

\begin{proof}
  The algorithm computes Breadth-First Search by Theorem~\ref{thm:abfs}. A
  processor is assigned to a nonzero in the sparse matrix representation, so
  there are $O(m)$ processors for the edges. At each step $k$ every processor
  reads its $x_k(j)$ value and performs the algebraic operations for the matrix
  product and writes out the new $x_{k+1}(i)$ value, specifically the processor
  performs,

  \begin{displaymath}
    x_{k+1}(i) \gets x_{k+1}(i) \oplus x_{k}(j) \otimes A(i,j).
  \end{displaymath}

  A processor writing the $x_{k+1}(i)$ value also updates $V_{k+1}$ so $i$ can
  be masked out of subsequent calculations. Hence only the submatrix
  $A[V_{k+1},V_k]$ is used in the computation at each step $k$. Observe that
  concurrent writes to the same $i$ index in $x_{k+1}$ do not change the result
  and so arbitrary write resolution suffices. The same holds for updating $i$ in
  $V_{k+1}$.

  Under the Work-Depth (WD) model there are as many processors as needed to
  compute a single step $k$ in $O(1)$ time. Here the processors perform just the
  necessary calculations to produce the next submatrix each step, hence the work
  is the sum of actual operations that are performed over all the steps. Each
  algebraic operation takes $O(1)$ time, then from Theorem~\ref{thm:abfs_ops} it
  takes $O(n)$ time in total. It follows from Lemma~\ref{lem:abfs_nnz} that at
  most $m$ nonzeros are referenced. Although only the first nonzero of a row is
  needed, each processor must evaluate if its current nonzero element is in
  $V_{k+1}$. The work is then $O(m)$ and it follows from Brent's Theorem that it
  takes $O(m/p + t)$ total time and $O(m+pt)$ work on a PRAM. Therefore it takes
  $O(t)$ time and $O(m)$ work using $O(m/t)$ processors on a PRAM.
\end{proof}

Our Algorithm~\ref{alg:abfs_spmv} is an optimal algebraic BFS algorithm for
sparse graphs that is deterministic and does not depend on ordering or lack of
ordering in the matrix and vector indices. Although this algorithm fills a gap
in the study of BFS, it offers no theoretical advantage over the simple
combinatorial algorithm. However, we believe it offers a practical advantage
over a theoretically optimal SpMmSpV-BFS because it requires $O(n)$ instead of
$O(m)$ algebraic operations. Our technique of hiding or masking portions of $A$
and the input vector could benefit new algebraic graph libraries that already
feature masked sparse linear algebra operations~\cite{bib:davis2019,
  bib:yang2019, bib:yang2018, bib:buluc_graphblas2017}, specifically for
SpMmSpV. Since our method uses progressively smaller submatrices of $A$ each
step, it could prove useful in lowering the communication cost in distributed
computing methods that bitmap and distribute partitions of the adjacency
matrix~\cite{bib:ueno2017}. Next we'll demonstrate how to easily integrate our
method in the popular Compressed Sparse Row (CSR) format which is used in many
sparse matrix libraries~\cite{bib:sparse++, bib:oski, bib:intelmkl}.

\section{\label{sec:sequential}Practical optimal sequential algorithm}
The CSR format is a well-known sparse matrix representation that utilizes three
arrays, $nz$, $col$, and $row$, to identify the nonzero elements. The $nz$ array
holds the nonzero values in row-major order. The $col$ array contains the column
indices for nonzeros in the same row-major order of $nz$, and $row$ is an array
of $col$ indices for the first nonzero in each row of the matrix. The last value
in $row$ must be one more than the last $col$ index. A matrix-vector
multiplication in CSR iterates over the gap between successive values of the
$row$ array to access each nonzero in a row. Recall we can use $A^T$ in the
matrix-vector multiplications since it works for to both directed and undirected
graphs. So we give the basic CSR matrix-vector multiplication for $\vec{y} \gets
A^T\vec{x}$ in Listing~\ref{lst:csr_matvec}.

\setcounter{algorithm}{0}
\begin{algorithm}[H]
\floatname{algorithm}{Listing}
\caption{\label{lst:csr_matvec}Transpose Matrix-vector multiplication in CSR}
\begin{algorithmic}
\Require $nz$,$col$,$row$ \Comment CSR data structures
\Require $\vec{x},\vec{y}$ \Comment input and output vectors
\For{$i=1..n$}
  \For{$j = row[i] .. row[i+1]-1$}
    \State $y(col[j]) \gets y(col[j]) \oplus nz[j] \otimes x(i)$
  \EndFor
\EndFor
\end{algorithmic}
\end{algorithm}

Our Algorithm~\ref{alg:abfs_csr} requires a simple modification to this CSR
matrix-vector multiplication. We add an array, $T$, to store nonzero indices for
vertices that have been visited. This is used to limit the multiplication over
the appropriate submatrix of $A$ every step. At each step we iterate over only
nonzeros in $\vec{x}$, the indices of which are stored in another array $L$. We
could have used a sparse vector representation for $\vec{x}$ each step, but for
simplicity we just increment a pointer in $L$.

\begin{algorithm}[H]
\caption{\label{alg:abfs_csr}}
\begin{algorithmic}[1]
\Require $nz$,$col$,$row$ \Comment CSR data structures
\Require $\vec{x},\vec{y}$ \Comment input and output vectors
\Require $T,L$ \Comment arrays of size n
\Statex Initialize $\vec{x}$ and $L$ with source vertex
\State set $start := 0$ and $end := 1$ and $z := end$
\For{$k=1,2,\ldots$ until end of component}
  \For{$j=L[start] .. L[end]$}
    \For{$i=row[j] .. row[j+1]-1$}
      \If{$T[col[i]]$ is $0$}
        \State $y(col[i]) \gets y(col[i]) \oplus nz[i] \otimes x(j)$
        \State set $L[z] := col[i]$ and $T[col[i]] := 1$
        \State set $z := z + 1$
      \EndIf
    \EndFor
    \State set $T[j] := 1$ and $x(j) := 0$
  \EndFor
  \State set $start := end$ and $end := z$
  \State exchange pointers between $\vec{x}$ and $\vec{y}$
\EndFor
\end{algorithmic}
\end{algorithm}

Each column index $col[i]$ from $row[j]$ up to $row[j+1]$ is tested if it refers
to a vertex that has already been visited. If not, the CSR matrix-vector product
is performed. The $col[i]$ are added to $L$ and marked as visited in $T$. Each
new entry in the product vector $\vec{y}$ must be nonzero because only nonzero
operands are used in the multiplication.

Observe that each $col[i]$ is immediately marked \emph{visited}. This
effectively updates $V_{k+1}$ and simultaneously ensures that only the first
nonzero in a row is used at step $k$, hence there will be only $n$ unique
entries in the array $L$. Then by Theorem~\ref{thm:abfs_ops} the total number of
algebraic operations is $2(n-1)$ rather than $2m$. If desired, all $col[i]$ in a
step can contribute to the output by simply adding a test if $T[i]$ is zero
before the inner loop over the $row$ array, then removing the update to
$T[col[i]]$ and allowing $L$ to take $O(m)$ values.

Since Algorithm~\ref{alg:abfs_csr} is based on Algorithm~\ref{alg:abfs_spmv} and
does not add any new operations, then Theorem~\ref{thm:abfs_csr} follows
immediately.

\begin{theorem}
  \label{thm:abfs_csr}
  Algorithm~\ref{alg:abfs_csr} computes an algebraic Breadth-First Search in
  $O(m+n)$ time for sparse $G$.
\end{theorem}

It isn't difficult to see that Algorithm~\ref{alg:abfs_csr} is very similar to
the combinatorial algorithm. Algorithm~\ref{alg:abfs_csr} is algebraic in the
sense that it solves BFS by the matrix equations in Theorem~\ref{thm:abfs} using
any appropriate semiring for the addition and multiplication operations. Our
intent here is to demonstrate that existing sparse matrix methods require only a
simple adaptation to achieve the optimality of the combinatorial algorithm while
maintaining their practical advantages. Our main algorithmic result is given in
Algorithm~\ref{alg:abfs_spmv} since it captures the general concept given in
Theorem~\ref{thm:abfs} and therefore is amenable to many forms of sparse matrix
and sparse vector implementations.

\section{\label{sec:parallel}Practical work-optimal parallel algorithm}
We give parallel, work-optimal CSR algorithm in Algorithm~\ref{alg:abfs_pcsr}
that is both simple and practical. Aside from the use of \textbf{for all} over
the $L$ and $row$ arrays, a significant difference between this algorithm and
Algorithm~\ref{alg:abfs_csr} is that we have to avoid adding duplicates to the
$L$ array.

\begin{algorithm}[H]
\caption{\label{alg:abfs_pcsr}}
\begin{algorithmic}[1]
\Require $nz$,$col$,$row$ \Comment CSR data structures
\Require $\vec{x},\vec{y}$ \Comment input and output vectors
\Require $T,L,P$ \Comment arrays of size n
\Statex Initialize $\vec{x}$ and $L$ with source vertex
\State set $start := 0$ and $end := 1$ and $z := end$
\For{$k=1,2,\ldots$ until end of component}
  \ForAll{$j=L[start] .. L[end]$}
    \ForAll{$i=row[j] .. row[j+1]-1$}
      \If{$T[col[i]]$ is $0$}
        \State $y(col[i]) \gets y(col[i]) \oplus nz[i] \otimes x(j)$
        \label{alg:abfs_pcsr_write}
        \State set $P[col[i]] := i$
        \If{$P[col[i]]$ is $i$}
          \State set $L[z] := col[i]$ and $T[col[i]] := 1$
          \State set $z := z + 1$
        \EndIf
      \EndIf
    \EndFor
    \State set $T[j] := 1$ and $x(j) := 0$
  \EndFor
  \State set $start := end$ and $end := z$
  \State exchange pointers between $\vec{x}$ and $\vec{y}$
\EndFor
\end{algorithmic}
\end{algorithm}

We use the PRAM model here so we can focus on the main aspects of the algorithm
without adding the complexity of how physical machines and parallel programming
manage write conflicts. Later, we'll comment on how we handle synchronization
for our practical implementation of Algorithm~\ref{alg:abfs_pcsr}.

In a PRAM all processors execute instructions simultaneously and can
concurrently access any memory cell from a global pool of memory. During the
execution of Algorithm~\ref{alg:abfs_pcsr}, two processors at level $k$ of the
BFS tree can find the same new vertex $u=col[i]$ at level $k+1$. Concurrent
write to $y(col[i])$ at line~\ref{alg:abfs_pcsr_write} does not pose a problem
because the same value is written by any processor. But we have to avoid adding
$u$ to $L$ more than once. This can be handled with a new parent array $P$ where
processors concurrently write to $P[u]$ their frontier vertex $v$ that reached
$u$. It doesn't matter which processor wins the write to $P[u]$ so any write
resolution protocol suffices. In the next clock cycle, each processor reads
$P[u]$ and only the processor with a matching parent vertex for $u$ can add $u$
to $L$, hence avoiding duplicates.

Algorithm~\ref{alg:abfs_pcsr} has the same work and depth as
Algorithm~\ref{alg:abfs_spmv} so Theorem~\ref{thm:abfs_pcsr} follows.

\begin{theorem}
  \label{thm:abfs_pcsr}
  Algorithm~\ref{alg:abfs_pcsr} computes an algebraic Breadth-First Search in
  $t$ steps and $O(m)$ work using $O(m/t)$ PRAM processors for sparse $G$.
\end{theorem}

A practical implementation of Algorithm~\ref{alg:abfs_pcsr} will need to handle
read and write conflicts on a specific machine architecture and parallel
programming methodology. We created a multithreaded implementation of
Algorithm~\ref{alg:abfs_pcsr} using OpenMP. Atomics are used for reading and
updating the $T$ array. To avoid adding duplicates to $L$, a local array is kept
for each thread where it can store the new vertices discovered from its subset
of frontier vertices. Each thread iterates over these vertices in its local
array and atomically updates $T$ to remove any that had been discovered by other
threads. This eliminates all duplicates so only unique vertices can be added to
the next frontier. Since there can only be $O(m)$ duplicates overall then our
OpenMP implementation remains asymptotically work-optimal.

Each thread must update $L$ with its final subset of new frontier vertices. To
avoid synchronization on $L$, each thread is given a unique block of space in
$L$ to which it can concurrently add its frontier subset without write
conflicts. This can be achieved without wasting space as follows. Each thread
concurrently writes the count of its frontier vertices to a new global array
indexed by thread IDs. Since each thread has a unique ID then no synchronization
is needed. All threads wait until this array has values from every thread. Then
in parallel, each thread sums up the values in this global array from the first
thread ID to its own. This gives each thread a unique starting offset in $L$ to
write their data and there is no overlap or gaps, thereby achieving full
concurrency on $L$.

\section{\label{sec:experiment}Experiments}
We will show that Theorem~\ref{thm:abfs} leads to significant savings in
algebraic operations on real-world graphs. Then we'll compare the sequential and
parallel runtime performance of our approach with that of the GraphBLAS
library. We begin by comparing Algorithm~\ref{alg:abfs_csr} to simple CSR
implementations for BFS using dense vector (SpMV), sparse vector (SpMSpV), and
masked sparse vector (SpMmSpV). In each of the BFS implementations in our tests,
the visited vertices in a search are tracked in a similar manner as in
Algorithm~\ref{alg:abfs_csr}. In the next descriptions we denote these test
implementations as follows. Let SpMV-BFS denote the dense vector method,
SpMSpV-BFS for the sparse vector method, and finally SpMmSpV-BFS for a
``masked'' sparse vector method. In SpMV-BFS all nonzeros in $A$ are multiplied
each step. In SpMSpV-BFS we only count algebraic operations due to nonzeros in
the sparse vector. The SpMmSpV-BFS is nearly identical to
Algorithm~\ref{alg:abfs_csr} but tests visited frontier vertices after the
algebraic operations, and is therefore an asymptotically optimal ``masked''
sparse vector method.

Listed in Table~\ref{tbl:graphdata} are graphs from the Stanford Network
Analysis Project (SNAP)~\cite{bib:snapnets} used in the
experiments.\footnote{Diameters may differ from SNAP due to random sampling.} We
chose a source vertex for each graph such that BFS is run for the entire
reported diameter.\footnote{The number of BFS steps is one more than the
  diameter.} We count two algebraic operations for the
$(\oplus,\otimes)$-semiring operations in the inner loop of the CSR
multiplication. A comparison of the methods is illustrated by a log-scale
histogram in Figure~\ref{fig:op_histogram} where it is clear that
Algorithm~\ref{alg:abfs_csr} requires orders of magnitude fewer algebraic
operations than the other methods with the exception of the masked sparse vector
approach.

\begin{table}[t]
\caption{\label{tbl:graphdata} Test Graphs}
\centering
\begin{tabular}{lrrr}
  \toprule
  {} & n (vertices) & m (edges) & $D_G$ (diameter) \\
  \midrule
  roadNet-TX & 1,379,917 & 1,921,660 & 1057 \\
  roadNet-CA & 1,965,206 & 2,766,607 & 854 \\
  roadNet-PA & 1,088,092 & 1,541,898 & 787 \\ 
  com-Amazon & 334,863 & 925,872 & 44 \\
  com-Youtube & 1,134,890 & 2,987,624 & 20 \\
  com-LiveJournal & 3,997,962 & 34,681,189 & 17 \\
  com-Orkut & 3,072,441 & 117,185,083 & 9
\end{tabular}
\end{table}

\begin{figure}[t]
\centering
\resizebox{.55\textwidth}{.35\textwidth}{%
\input{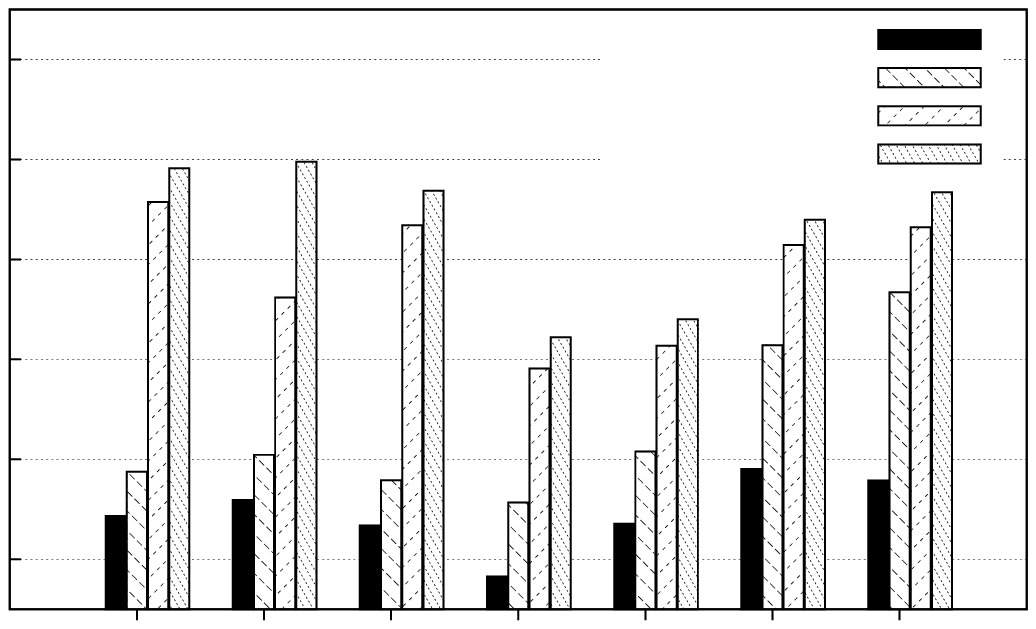}
}
\caption{\label{fig:op_histogram} Comparison of total algebraic operations.}
\end{figure}

The SpMmSpV-BFS should take $4m$ operations as we claimed in
Sections~\ref{sec:introduction} and~\ref{sec:submatrix}. At the worst end, the
SpMV-BFS multiplies all $2m$ nonzeros in $A$ every step leading to $4m(D_G+1)$
algebraic operations. In contrast, Algorithm~\ref{alg:abfs_csr} multiplies $n-1$
nonzeros as asserted by Theorem~\ref{thm:abfs_ops}, leading to $2(n-1)$
algebraic operations in total. All of this bears out in the experimental results
listed in Table~\ref{tbl:snap_ops}. A simple modification to
Algorithm~\ref{alg:abfs_csr} was described in Section~\ref{sec:sequential} that
would include all nonzeros in a row, then according to Lemma~\ref{lem:abfs_nnz}
our optimal algebraic BFS should take $2m$ algebraic operations on a sparse
graph. We tested this modification in the experiments and verified that it
indeed results in $2m$ algebraic operations on each of the graphs; therefore a
straightforward implementation at worst takes half the number of operations as
an optimal SpMmSpV-BFS.

\begin{table*}[t]
\caption{\label{tbl:snap_ops} Total Algebraic Operations}
\centering
\begin{tabular}{lrrrr}
  \toprule
  {} & Algorithm~\ref{alg:abfs_csr} & SpMmSpV-BFS & SpMSpV-BFS & SpMV-BFS \\
  \midrule
  roadNet-TX & 2,702,272 & 7,516,804 & 3,739,129,896 & 8,132,465,120 \\
  roadNet-CA & 3,914,052 & 11,041,552 & 4,134,833,458 & 9,461,795,940 \\
  roadNet-PA & 2,175,122 & 6,166,056 & 2,192,455,678 & 4,860,062,496 \\ 
  com-Amazon & 669,724 & 3,703,488 & 80,934,254 & 166,656,960 \\
  com-Youtube & 2,269,778 & 11,950,496 & 136,126,648 & 250,960,416 \\
  com-LiveJournal & 7,995,922 & 138,724,756 & 1,391,162,196 & 2,497,045,608 \\
  com-Orkut & 6,144,880 & 468,740,332 & 2,093,464,344 & 4,687,403,320
\end{tabular}
\end{table*}

We compared our algorithm implementations against the SuiteSparse GraphBLAS
library~\cite{bib:davis2019, bib:suitesparsegraphblas, bib:davis2018}. This is
considered a complete reference implementation and for convenience we will refer
to it simply as GraphBLAS. This library implements SpMmSpV-BFS using a masked
sparse vector~\cite[c.f. UserGuide pg.~191]{bib:suitesparsegraphblas}. We use
version 2.2.2 and 3.0.1 of GraphBLAS for the sequential and parallel tests,
respectively. The parallel BFS implementations of both our method and GraphBLAS
use OpenMP multithreading. The experiments were run on a single workstation with
over 256 GB of RAM and 28 Intel Xeon E5-2680 cores. Each BFS starts at vertex
$0$ so the absolute diameter is less than those reported in
Table~\ref{tbl:graphdata}. We use $T_{GB}$ to denote the runtime for GraphBLAS
and $T_{Alg}$ for the runtime of either Algorithm~\ref{alg:abfs_csr} or
Algorithm~\ref{alg:abfs_pcsr}. In the plots to follow the performance comparison
is given as the ratio $T_{GB}/T_{Alg}$ so if $T_{Alg}$ is faster then the ratio
is greater than one and will appear higher on the y-axis. In all experiments our
sequential and parallel implementations were faster than GraphBLAS, often by
over an order of magnitude.

A comparison of sequential runtime is given in
Figure~\ref{fig:serial_time_histogram}. Our BFS is 19-24x faster than GraphBLAS
on the road network graphs, and on average it is 22x faster. These graphs have
large diameter so inefficiencies are compounded with each step. The gains for
the lower diameter graphs are more modest but on average we are still more than
3.5x faster. In these graphs the work is more concentrated within each level of
the BFS rather than distributed over the total number of levels. Memory cache
effects could boost performance in these graphs more than the large diameter
graphs where there is considerably less work in each level of the BFS
tree. Although our approach requires signficantly fewer algebraic computations,
we suspect memory access is the bottleneck and thus accounts for lower than
expected improvement over GraphBLAS on the larger, low-diameter graphs.

We remark that our OpenMP implementation of Algorithm~\ref{alg:abfs_pcsr}
performs more algebraic operations than the sequential implementation of
Algorithm~\ref{alg:abfs_csr} because threads can concurrently produce the same
frontier vertices in the matrix-vector multiplication before these duplicates
are removed. This and the overhead of managing parallel and sequential regions
account for the gap in performance with respect to the experimental results for
the sequential algorithm. But since there are $O(m)$ duplicates overall, the
implementation remains asymptotically work-optimal. We remind the reader that
details on our implementation are at the end of Section~\ref{sec:parallel}.

The parallel runtime is given in Figure~\ref{fig:parallel_time} where it is also
evident that our method is considerably faster than GraphBLAS. In two-thirds of
the tests our algorithm completed in less than one-tenth of a second, and took
over one second for just one test where two threads were used on the largest
graph, com-Orkut, which took 1.59 seconds. At 64 threads on the two largest
graphs, com-Orkut and com-LiveJournal, our BFS completes in 0.153 and 0.082
seconds resulting in 11x and 17x faster time than GraphBLAS, respectively. For
com-Orkut we achieve over 1.5 GTEPS (Giga Traversed Edges Per Second). Also
notable on these large graphs is that our performance with respect to GraphBLAS
scales linearly on the whole. The exception is a drop-off with com-Amazon at 64
threads. We don't have an explanation for this drop but both our method and
GraphBLAS suffered it. For the large diameter graphs our performance peaks with
fewer threads, again because the per level work is very low so the overhead of
adding more threads becomes significant. Since we only have 28 physical cores,
running 64 virtual threads on these graphs with very few edges in each BFS level
may compound cache misses. But we are still about 9-12x faster than GraphBLAS at
our peak for these graphs. The raw wallclock timings are available in
Appendix~\ref{sec:timing}.

\begin{figure}[t]
\centering
\resizebox{.55\textwidth}{.35\textwidth}{%
\input{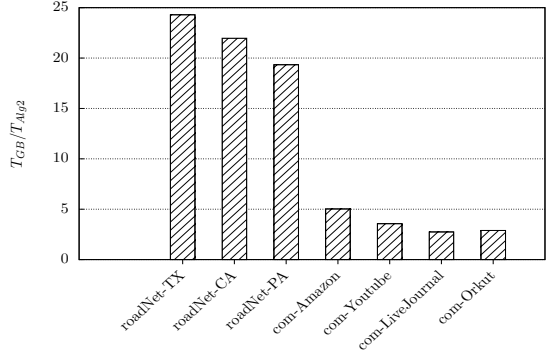}
}
\caption{\label{fig:serial_time_histogram} Comparison of sequential runtime.}
\end{figure}

\begin{figure*}[t]
\centering
\begin{subfigure}[t]{.45\textwidth}
\resizebox{1\textwidth}{.75\textwidth}{%
  \input{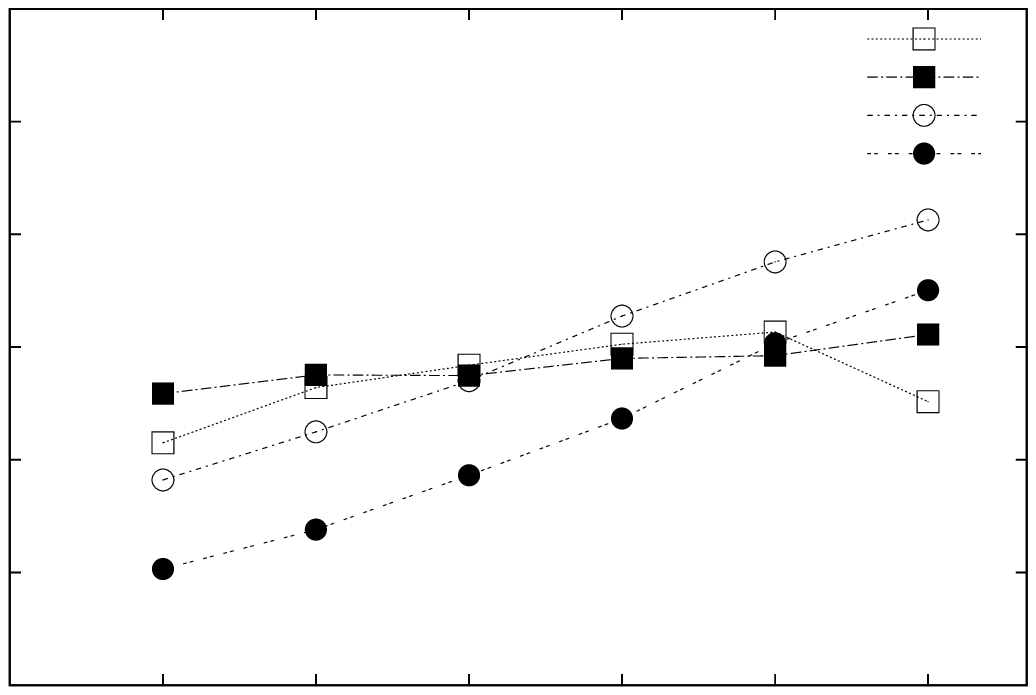}
}
\end{subfigure}
\begin{subfigure}[t]{.45\textwidth}
\resizebox{1\textwidth}{.75\textwidth}{%
  \input{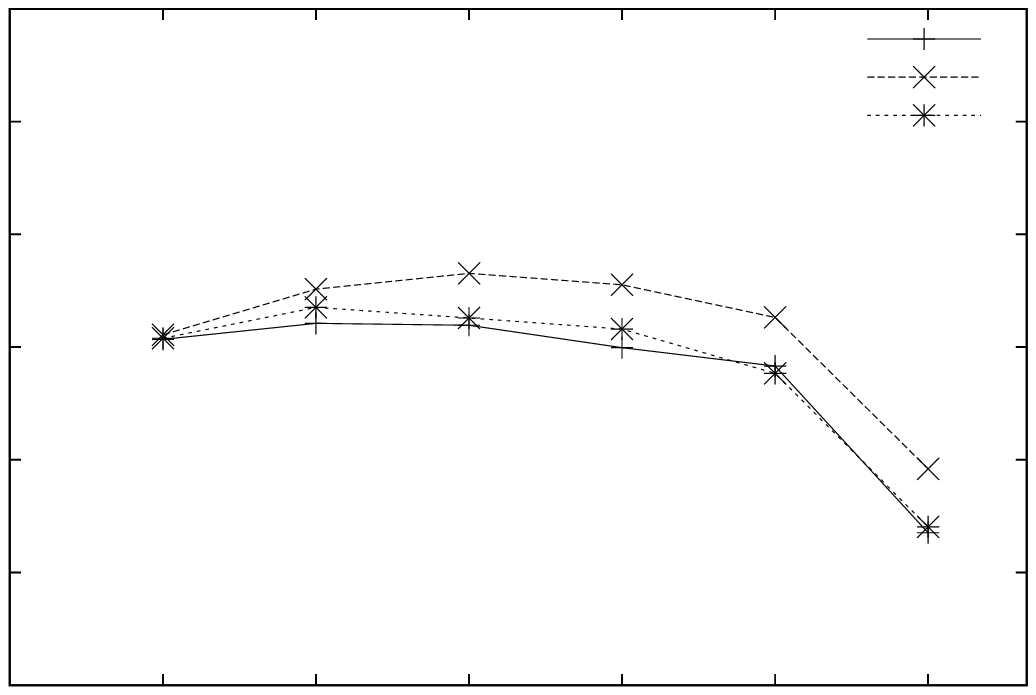}
}
\end{subfigure}
\caption{\label{fig:parallel_time} Comparison of parallel runtime.}
\end{figure*}

\section{\label{sec:conclusion}Conclusion}
We introduced a new algebraic formulation for Breadth-First Search that ensures
optimal work on sparse graphs by multiplying progressively smaller submatrices
of the adjacency matrix $A$. We show that BFS can be computed by the recurrence
$\vec{x}_{k+1}=A[V_{k+1},V_k]\vec{x}_k$, where masking row and column vectors in
$A$ corresponding to frontier nonzeros in each input vector will prevent
multiplying the same frontier nonzeros. Our submatrix multiplication approach
takes $O(n)$ instead of the $O(m)$ algebraic operations of a theoretically
optimal SpMmSpV-BFS on both undirected and directed graphs.

We gave a general algorithm using our submatrix formulation for BFS, showing it
is work-optimal for both sequential and parallel processing. We demonstrated it
is easily amended with CSR sparse-matrix multiplication and verified it leads to
significant performance improvement over the current state-of-the-art
SpMmSpV-BFS. We believe our approach can be easily integrated with other
existing matrix methods, and benefit those that support masking operations such
as the GraphBLAS.

Our paper closes a gap between the linear algebraic and graph-theoretic solution
for BFS. The methods we introduced may be useful in devising linear algebraic
formulations of other graph algorithms.

\section*{Acknowledgments}
The author is grateful to David G. Harris and Christopher H. Long for their
helpful comments. The author also thanks the anonymous reviewers for their
critical suggestions that improved the paper.

\bibliographystyle{abbrv}
\bibliography{optimal_algebraic_bfs}

\appendix
\section{\label{sec:lineartransform}Linear transformation}
We begin with the following definitions and claims.

\begin{definition}
  \label{def:submatrix}
  Let $A_k \in \{0,1\}^{n\times n}$ be a
  Boolean symmetric matrix such that $A_k(i,j)=1$ if $A(i,j)=1$ and $i,j \in
  V_k$, and is zero otherwise.
\end{definition}

  The size of $A_k$ does not change at each step, only the $A_{*,i},A_{i,*}$ are
  zeroed out for those $i$ not in $V_k$. Then $A_1$ is the adjacency matrix $A$.

\begin{definition}
  \label{def:selection_matrix}
  Let $S_k \in \{0,1\}^{n\times n}$ be a \emph{selection matrix} that is a
  Boolean diagonal matrix such that $S_k(i,j)=1$ if $i=j$ and $i,j \in V_k$, and
  is zero otherwise. Since it is diagonal with some elements being zero, $S_k$
  is therefore an idempotent matrix.
\end{definition}

A \emph{selection matrix} is a Boolean diagonal matrix that is used to mask or
zero out rows/columns of some other matrix. A selection matrix $S$ that is not
the Identity will have $\vec{0}$ for some row and column vectors. Then by the
usual rules of matrix multiplication, multiplying a matrix $A$ on the right by
$S$ will inherit the $\vec{0}$ column vectors in $S$, and multiplying $A$ on the
left by $S$ inherits the $\vec{0}$ rows in $S$. A symmetric selection on a
matrix $A$ is then given by $S A S$, which returns a new matrix with the same
dimensions of $A$ containing only the $A_{*,i},A_{i,*}$ corresponding to nonzero
$S_k(i,i)$ diagonal elements, and all other rows/columns are zeroed. For
example, if all diagonal elements in $S$ were one except for $S(2,2)$, then $S A
S$ returns $A$ with $A_{*,2}$ and $A_{2,*}$ as the zero vector.

\begin{claim}
  \label{clm:selection_identity}
  The equality $S_{k+1}S_k=S_{k+1}$ holds for $k\ge 1$.
\end{claim}

\begin{proof}
  By Definition~\ref{def:selection_matrix} the nonzeros in $S_{k+1}$ must be in
  $S_k$ and so these rows are identical. Then multiplying $S_k$ on the left by
  $S_{k+1}$ annihilates the rows in $S_{k}$ indexed by $\vec{0}$ row vectors in
  $S_{k+1}$, hence the product $S_{k+1}S_k$ must return $S_{k+1}$.
\end{proof}

\begin{claim}
  \label{clm:submatrix_sequence}
  The equality $A_{k+1} = S_{k+1} A_k S_{k+1}$ holds for $k\ge 1$.
\end{claim}

\begin{proof}
  We prove this by induction. In the base step $A_2 = S_2 A_1 S_2$ holds because
  $S_2(i,i)$ is zero for the source vertex $i$ and hence symmetric selection
  annihilates $A_{*,i}$ and $A_{i,*}$ to give $A[V_2,V_2]=A_2$ which satisfies
  Definition~\ref{def:submatrix}. Then $S_3 A_2 S_3$ gives $A[V_3,V_3]=A_3$ by
  the same definition.

  Now assume $A_k = S_k A_{k-1} S_k$ is true for all steps $1..k$. Since
  $S_{k+1}$ masks out the nonzeros from $\vec{x}_{k+1}$ and all previous
  $\vec{x}_k$ have already been masked in $A_k$, then $S_{k+1} A_k S_{k+1}$
  gives $A_{k+1}=A[V_{k+1},V_{k+1}]$.
\end{proof}

\begin{claim}
  \label{clm:submatrix_identity}
  The equality $A_{k}=S_{k} A S_{k}$ holds for $k\ge 1$.
\end{claim}

\begin{proof}
  We prove this by induction. In the base step the claim follows trivially for
  $A_1=S_1 A S_1$ since $S_1$ is the Identity. Now in the inductive step, assume
  $A_k = S_k A S_k$ holds. Then applying Claims~\ref{clm:selection_identity}
  and~\ref{clm:submatrix_sequence} gives $A_{k+1} = S_{k+1} A S_{k+1}$ as
  follows.

  \begin{align*}
    A_{k+1} &= S_{k+1} A_k S_{k+1}
    \tag{Claim~\ref{clm:submatrix_sequence}} \\
    &= S_{k+1} (S_k A S_k) S_{k+1}
    \tag{Induction} \\
    &= S_{k+1} A S_{k+1}
    \tag{Claim~\ref{clm:selection_identity}}
  \end{align*}
\end{proof}

\newtheorem*{prop:transform}{Proposition \ref{prop:linear_transformation}}
\begin{prop:transform}
There is a linear transformation on $\vec{y}_k = A\vec{y}_{k-1}$ that gives
$\vec{x}_{k+1}=A[V_k,V_k]\vec{x}_k$, specifically $\vec{x}_{k+1} = A_k\vec{y}_k$
and subsequently $\vec{x}_{k+1}=A_kA^{k-1}\vec{x}_1$.
\end{prop:transform}

\begin{proof}
  We first show that $S_k \vec{x}_k$ is equal to $S_k \vec{y}_k$. Here
  $\vec{y}_k$ and $\vec{x}_k$ contain nonzeros that have not been produced by
  previous steps. It follows from Theorem~\ref{thm:abfs} that $\vec{x}_k$
  contains only such nonzeros. Now $S_k \vec{y}_k$ annihilates the nonzeros in
  $\vec{y}_k$ that are not in $V_k$, hence $S_k \vec{y}_k$ is equal to
  $\vec{x}_k$. Since $S_k$ is idempotent then $S_k \vec{x}_k = S_k (S_k
  \vec{y}_k) = S_k \vec{y}_k$. Using this equality and
  Claim~\ref{clm:submatrix_identity} we can show that $A_k \vec{y}_k$ is equal
  to $A_k \vec{x}_k$.

  \begin{align*}
    A_k \vec{y}_k &= S_k A S_k \vec{y}_k
    \tag{Claim~\ref{clm:submatrix_identity}} \\
    &= S_k A S_k \vec{x}_k \\
    &= A_k \vec{x}_k \tag{Claim~\ref{clm:submatrix_identity}} \\
  \end{align*}

  This leads to $\vec{x}_{k+1} = A_k \vec{x}_k = A_k \vec{y}_k$. Iteration on
  $\vec{y}_k=A\vec{y}_{k-1}$ yields $\vec{y}_k=A^{k-1}\vec{y}_1$. Since the
  source vector is the same for this conventional recurrence and that of
  Theorem~\ref{thm:abfs}, then $\vec{x}_1=\vec{y}_1$, giving the result
  $\vec{x}_{k+1}=A_k\vec{y}_k$ and $\vec{x}_{k+1}=A_k A^{k-1} \vec{x}_1$ as
  claimed.
\end{proof}

We emphasize that the result of Proposition~\ref{prop:linear_transformation}
supposes that a chosen semiring is applied consistently. If the arithmetic
semiring was used to produce $\vec{x}_{k+1}$ then it must be used to compute
$A_kA^{k-1}\vec{x}_1$.

\section{\label{sec:timing}Parallel algorithm timings}
Tables~\ref{tbl:alg_timing} and~\ref{tbl:graphblas_timing} list the wallclock
timings in seconds for the parallel runtime experiments plotted in
Figure~\ref{fig:parallel_time} of Section~\ref{sec:experiment}.

\begin{table}[H]
\caption{\label{tbl:alg_timing} Algorithm~\ref{alg:abfs_pcsr} OpenMP
  wallclock(s)}
\centering
\footnotesize
\begin{tabular}{lrrrrrrr}
  \toprule
{} & \tiny{roadNet-TX} & \tiny{roadNet-CA} & \tiny{roadNet-PA} &
\tiny{com-Amazon} & \tiny{com-Youtube} &
\tiny{com-LiveJournal} & \tiny{com-Orkut} \\
\midrule
Threads \\
\midrule
2 & 0.078264 & 0.109803 & 0.062554 & 0.040858 & 0.079981 & 0.758120 & 1.590389\\
4 & 0.068324 & 0.076542 & 0.049612 & 0.023068 & 0.054668 & 0.427502 & 0.904479\\
8 & 0.065665 & 0.067798 & 0.054572 & 0.015790 & 0.035317 & 0.232289 & 0.498732\\
16 & 0.075797 & 0.071152 & 0.057704 & 0.011845 & 0.025740 & 0.132911 & 0.274236\\
32 & 0.085144 & 0.083455 & 0.071895 & 0.011375 & 0.023407 & 0.092609 & 0.182314\\
64 & 0.285308 & 0.234109 & 0.219250 & 0.023498 & 0.027105 & 0.081694 & 0.153125
\end{tabular}
\end{table}

\begin{table}[H]
\caption{\label{tbl:graphblas_timing} GraphBLAS OpenMP
  wallclock(s)}
\centering
\footnotesize
\begin{tabular}{lrrrrrrr}
  \toprule
{} & \tiny{roadNet-TX} & \tiny{roadNet-CA} & \tiny{roadNet-PA} &
\tiny{com-Amazon} & \tiny{com-Youtube} &
\tiny{com-LiveJournal} & \tiny{com-Orkut} \\
\midrule
Threads \\
\midrule
2 & 0.656726 & 0.947223 & 0.528064 & 0.181383 & 0.480276 & 2.678422 & 3.249154\\
4 & 0.632048 & 0.873912 & 0.506795 & 0.143786 & 0.368506 & 2.030834 & 2.353997\\
8 & 0.600930 & 0.853390 & 0.522071 & 0.112992 & 0.237032 & 1.511427 & 1.812888\\
16 & 0.604373 & 0.834515 & 0.515823 & 0.096466 & 0.191989 & 1.286389 & 1.413317\\
32 & 0.605746 & 0.800793 & 0.489050 & 0.099752 & 0.177613 & 1.250459 & 1.492713\\
64 & 0.728695 & 0.885137 & 0.580791 & 0.134360 & 0.234332 & 1.429025 & 1.737067
\end{tabular}
\end{table}
\end{document}